\documentclass[11pt]{article}
\pdfoutput=1
\usepackage[activate={true,nocompatibility},final,kerning=true]{microtype}
\usepackage{amsmath,amssymb,amsfonts,amsthm}
\usepackage{graphicx}
\usepackage{fullpage}
\usepackage[x11names]{xcolor}
\usepackage[pdfusetitle,backref,colorlinks,citecolor=RoyalBlue4, linkcolor=RoyalBlue4]{hyperref}
\usepackage{cleveref} 
\usepackage[english]{babel}
\addto\extrasenglish{%
}

\usepackage{color}
\usepackage{wrapfig}
\usepackage{tikz}
\usetikzlibrary{decorations.pathreplacing}
\usepackage{setspace}
\usepackage{algorithm}
\usepackage[noend]{algpseudocode}
\usepackage[framemethod=tikz]{mdframed}
\usepackage{xspace}
\usepackage{pgfplots}
\usepackage{framed}
\usepackage{subcaption}
\usepackage{thmtools}
\usepackage{thm-restate}
\pgfplotsset{compat=1.5}

\newtheorem{theorem}{Theorem}[section]
\newtheorem{corollary}[theorem]{Corollary}
\newtheorem{lemma}[theorem]{Lemma}

\newtheorem{definition}[theorem]{Definition}

\numberwithin{equation}{section}

\newenvironment{proofof}[1]{\begin{trivlist} \item {\bf Proof
#1:~~}}
  {\qed\end{trivlist}}

\newcommand{\namedref}[2]{\hyperref[#2]{#1~\ref*{#2}}}

\newcommand{\alglab}[1]{\label{alg:#1}}
\renewcommand{\algref}[1]{\namedref{Algorithm}{alg:#1}}

\def \YES    {\mdef{\mathsf{YES}}}
\def \NO    {\mdef{\mathsf{NO}}}
\def \IC    {\mdef{\mathsf{IC}}}
\def \biasdetect    {\mdef{\textsc{BiasDetect}}}
\def \privmed    {\mdef{\textsc{PrivMed}}}

\newcommand{\diffdist}{\textsc{DiffDist}}

\newcommand{\discpred}{{\textsc{DiscPred}}}

\newcommand{\PPr}[1]{\ensuremath{\mathbf{Pr}\left[#1\right]}}
\newcommand{\PPPr}[2]{\ensuremath{\underset{#1}{\mathbf{Pr}}\left[#2\right]}}

\newcommand{\EEx}[2]{\ensuremath{\underset{#1}{\mathbb{E}}\left[#2\right]}}
\renewcommand{\O}[1]{\ensuremath{O\left(#1\right)}}
\newcommand{\tO}[1]{\ensuremath{\widetilde{O}\left(#1\right)}}
\newcommand{\eps}{\varepsilon}

\newcommand{\shortdash}{\text{-}}
\def \TV {\text{TV}}

\def \calA    {\mdef{\mathcal{A}}}

\def \calD    {\mdef{\mathcal{D}}}

\def \calX    {\mdef{\mathcal{X}}}
\def \calY    {\mdef{\mathcal{Y}}}

\renewcommand{\labelenumi}{(\arabic{enumi})}

\newcommand{\mdef}[1]{{\ensuremath{#1}}\xspace}  

\DeclareMathOperator*{\poly}{poly}

\newcommand{\E}[2][]{\mdef{\underset{#1}{\mathbb{E}}\left[#2\right]}} 

\newcommand{\ignore}[1]{}

\newif\ifnotes\notestrue 
\ifnotes
\newcommand{\samson}[1]{\textcolor{purple}{{\bf (Samson:} {#1}{\bf ) }} \marginpar{\tiny\bf
             \begin{minipage}[t]{0.5in}
               \raggedright S:
            \end{minipage}}}            							
\else
\newcommand{\samson}[1]{}
\fi

\newif\ifnotes\notestrue 
\ifnotes
\newcommand{\fred}[1]{\textcolor{purple}{{\bf (Fred:} {#1}{\bf ) }} \marginpar{\tiny\bf
             \begin{minipage}[t]{0.5in}
               \raggedright S:
            \end{minipage}}}            							
\else
\newcommand{\fred}[1]{}
\fi

\providecommand{\email}[1]{\href{mailto:#1}{\nolinkurl{#1}\xspace}}

\begin{document}

\title{Streaming Algorithms for Learning with Experts: Deterministic Versus Robust}
\author{David P. Woodruff \\  CMU \\ \texttt{ dwoodruf@cs.cmu.edu}  \and Fred Zhang\thanks{Work done in part while interning at Google.} \\ UC Berkeley \\ \texttt{ z0@berkeley.edu} \\\and Samson Zhou\thanks{Work done in part while at Carnegie Mellon University.} \\   UC Berkeley and Rice University \\ \texttt{samsonzhou@gmail.com}}
\date{}
\maketitle

\begin{abstract}
In the online learning with experts problem, an algorithm must make a prediction about an outcome on each of $T$ days (or times), given a set of $n$ experts who make predictions on each day (or time). The algorithm is given feedback on the outcomes of each day, including the cost of its prediction and the cost of the expert predictions, and the goal is to make a prediction with the minimum cost, specifically compared to the best expert in the set. Recent work by Srinivas, Woodruff, Xu, and Zhou (STOC 2022) introduced the study of the online learning with experts problem under memory constraints. 

However, often the predictions made by experts or algorithms at some time influence future outcomes, so that the input is adaptively chosen. Whereas deterministic algorithms would be robust to adaptive inputs, existing algorithms all crucially use randomization to sample a small number of experts. 

In this paper, we study deterministic and robust algorithms for the experts problem. We first show a space lower bound of $\widetilde{\Omega}\left(\frac{nM}{RT}\right)$ for any deterministic algorithm that achieves regret $R$ when the best expert makes $M$ mistakes. Our result shows that the natural deterministic algorithm, which iterates through pools of experts until each expert in the pool has erred, is optimal up to polylogarithmic factors. On the positive side, we give a randomized algorithm that is robust to adaptive inputs that uses $\widetilde{O}\left(\frac{n}{R\sqrt{T}}\right)$ space for  $M=O\left(\frac{R^2 T}{\log^2 n}\right)$, thereby showing a smooth space-regret trade-off.
\end{abstract}

\newpage

\section{Introduction}
Online learning with experts is a problem of sequential prediction. 
On each of $T$ days (or times), an algorithm must make a prediction about an outcome, given a set of $n$ experts who make predictions on the outcome. 
The algorithm is then given feedback on the cost of its prediction and on the expert predictions for the current day. 
In the \emph{discrete prediction with experts problem}, the set of possible predictions is restricted to a finite set, and the cost is 0 if the prediction is correct, and 1 otherwise. 
More generally, the set of possible predictions need not be restricted but we assume the costs are restricted to be in a range $[0,\rho]$ for some fixed parameter $\rho>0$, with lower costs indicating better performances of the algorithm or experts. 
This process continues for the $T$ days (or times), after which the performance of the algorithm is compared to the performance of the best performing expert. 
More formally, the goal for the online learning with experts problem is often quantified by achieving the best regret, which is the difference between the total cost of the algorithm and the total cost of the best performing expert, i.e., the expert that incurs the least overall cost, amortized over the total number of days. 

A well-known folklore algorithm for handling the discrete prediction with experts problem is the weighted majority algorithm~\cite{LittlestoneW94}. 
The deterministic variant of the weighted majority algorithm simply initializes ``weights'' for all experts to $1$, downweights any incorrect expert on a given day, and selects the prediction supported by the largest weight of experts. 
The algorithm solves the discrete prediction with experts problem with $\O{M+\log n}$ total mistakes, where $M$ is the number of mistakes made by the best expert and the coefficient of $M$ hidden by the big Oh notation is approximately $2.41$, thus achieving regret $\O{M+\log n}$. 
More generally, a large body of literature has studied optimizations to the weighted majority algorithm, such as a randomized variant where the probability of the algorithm selecting each prediction is proportional to the sum of the weights of the experts supporting the prediction. 
The randomized weighted majority algorithm achieves regret $\O{\sqrt{{\log n}/{T}}}$~\cite{LittlestoneW94}, which has been shown to be information-theoretically optimal, up to a constant. 
There have subsequently been many follow-ups to the weighted and randomized weighted majority algorithms that achieve similar regret bounds, but improve in other areas. 
For example, on a variety of structured problems, such as online shortest paths, follow the perturbed leader~\cite{KalaiV05} achieves the same regret bound as  randomized weighted majority  but uses less runtime on each day (or time).
In addition, the multiplicative weights algorithm achieves the optimal $\sqrt{{\ln n}/{(2T)}}$ regret, with a tight leading constant~\cite{GravinPS17}. 
However, these classic algorithms use a framework that maintains the cumulative cost of each expert, which requires the algorithm to store $\Omega(n)$ bits of information across its runtime. 

\paragraph{Memory bounds.} 
Recently, \cite{SrinivasWXZ22} considered the online learning with experts problem when memory is a premium for the algorithm. 
On the hardness side, they showed that any algorithm achieving a target regret $R$ requires $\Omega\left(\frac{n}{R^2T}\right)$ space, which implies that any algorithm achieving the information-theoretic $\O{\sqrt{{\log n}/{T}}}$ regret must use near-linear space. 
On the other hand, for random-order streams in which the algorithm may receive the worst-case input, but then the order of the days is uniformly random, \cite{SrinivasWXZ22} gave a nearly matching randomized algorithm that uses $\tO{\frac{n}{R^2T}}$ space and $R=\Omega\left(\sqrt{\frac{\log^2 n}{T}}\right)$, i.e., nearly all values of regret down to the information-theoretic limit. 
Moreover, when the number of mistakes $M$ made by the best expert is small, i.e., $M=\O{R^2T}$, \cite{SrinivasWXZ22} gave a randomized algorithm that uses $\tO{\frac{n}{RT}}$ space for arbitrary-order streams, thus showing that the hardness of their lower bound originates from a setting where the best expert makes a large number of mistakes. 

Subsequently, \cite{PengZ22} considered the online learning with experts problem  when  the algorithm is limited to use memory sublinear in $n$. 
They introduced a general framework that achieves $o(T)$ regret using $o(n)$ memory, with a trade-off parameter between space and regret that obtains $O_n\left(T^{4/5}\right)$ regret with $\O{\sqrt{n}}$ space and $O_n\left(T^{0.67}\right)$ regret with $\O{n^{0.99}}$ space. 

\paragraph{Adaptive inputs and determinism.}
Up to now, the discussion has focused on an oblivious setting, where the input to the algorithm may be worst-case, but is chosen independently of the algorithm and its outputs. 
The online learning with experts problem is often considered in the adaptive setting, where the input to the algorithm is allowed to depend on previous outputs by the algorithm. 
Formally, we define the adaptive setting as a two-player game between an algorithm $\calD$ and an adversary $\calA$ that adaptively creates the input stream to $\calD$. 
The game then proceeds in days and on the $t$-th day:
\begin{enumerate}
\item
The adversary $\calA$ chooses the outputs of all experts on day $t$ as well as the outcome of day $t$, depending on all previous stream updates and all previous outputs from the algorithm $\calD$. 
\item
The outputs (i.e., predictions) of all experts are simultaneously given to the algorithm $\calD$, which updates its data structures, acquires a fresh batch $R_t$ of random bits, and outputs a predicted outcome for day $t$.
\item
The outcome of day $t$ is revealed to $\calD$, while the predicted outcome for day $t$ by $\calD$ is revealed to the adversary $\calA$.
\end{enumerate}
The goal of $\calA$ is to induce $\calD$ to make as many incorrect predictions as possible throughout the stream. 
It is clear that any deterministic algorithm for the online learning with experts problem will maintain the same guarantees in the adaptive model. 
Unfortunately, both the algorithms of \cite{SrinivasWXZ22} and \cite{PengZ22} are randomized procedures that rely on iteratively sampling ``pools'' of experts, which can potentially be exploited by an adaptive adversary who learns the experts sampled in each pool. 
Interestingly, both the randomized weighted majority algorithm~\cite{LittlestoneW94} and the multiplicative weights algorithm~\cite{GravinPS17} are known to be robust to adaptive inputs. 

\subsection{Our Contributions}
In this paper, we study the capabilities and limits of sublinear space algorithms for the online learning with experts problem on adaptive inputs. 

\paragraph{Tight bounds for deterministic algorithms.}
First, we  provide a simple deterministic algorithm that uses space $\tO{\frac{nM}{RT}}$. 
Consider an algorithm that iteratively selects the next pool of $k=\tO{\frac{nM}{RT}}$ experts and running the deterministic majority algorithm on the experts in the pool, while removing any incorrect experts from the pool until the pool is completely depleted, at which point the next pool of $\tO{\frac{nM}{RT}}$ experts is selected. 
The main intuition is that each pool can incur at most $\O{\log n}$ mistakes before it is completely depleted and the best expert can only make $M$ mistakes. 
By the time the pool has cycled through $nM$ experts, i.e., $M$ times for each of the $n$ experts, then the best expert no longer makes any mistakes and will be retained by the pool. 
Thus, the total number of mistakes made by the deterministic algorithm is $\frac{nM}{k}\cdot\O{\log n}$. 
Hence, for a target average regret $R$, the total number of mistakes by the algorithm must be at most $M+RT\ge RT$, so it suffices to set $k=\tO{\frac{nM}{RT}}$ to achieve regret $R$. 
Since the algorithm runs deterministic majority on a pool of $k=\tO{\frac{nM}{RT}}$ experts, then this algorithm uses $\tO{\frac{nM}{RT}}$ space. 
However, for $M=\Omega({RT})$, the algorithm must use space that is near-linear in the number of experts $n$, which is undesirable when $n$ is large. 
(For  a detailed formal argument, see \autoref{sec:det-alg})

Therefore, it is natural to ask whether there exists a deterministic algorithm that is more space-efficient than this straighforward approach. 
Unfortunately, we first show that this is not the case:
\begin{theorem}[Memory lower bound for deterministic algorithms; also see \autoref{thm:lb:random}]\label{thm:main-intro}
For $n=o(2^T)$, any deterministic algorithm that achieves $R$ regret for the discrete prediction with experts problem must use $\Omega\left(\frac{nM}{RT}\right)$ space when the best expert makes $M$ mistakes.
\end{theorem}
Taken together with the deterministic procedure above, this resolves the deterministic streaming complexity of online learning with experts.

At a conceptual level, our lower bound in Theorem~\ref{thm:main-intro} shows that surprisingly, the number $M$ of the mistakes made by the best expert is an intrinsic parameter that governs the abilities and limitations of deterministic algorithms in this model. 
In fact, we show a stronger result in~\autoref{thm:lb:random} that any randomized algorithm that succeeds with probability at least $1-\exp(-T)$ must use $\Omega\left(\frac{nM}{RT}\right)$ space when the best expert makes $M$ mistakes. 

Moreover, we give an alternative proof in the regime when $M = \Omega(T)$. The proof differs from the proof of \autoref{thm:main-intro}. Instead, it leverages the communication complexity of  a new set disjointness problem, recently proposed by \cite{kamath2021simple}.     The statement is technically weaker \autoref{thm:main-intro}, and  appears in the appendix; see \autoref{sec:alternative}.

\paragraph{Overview of the proof of \autoref{thm:main-intro}.}
To prove the theorem, we consider the communication problem of  \(\eps\shortdash\diffdist\). 
It  combines $n$ instances of the distributed detection problem from \cite{braverman2016communication} and was first proposed by the prior work of  \cite{SrinivasWXZ22} to prove space lower bounds  for expert learning in random order stream.

Specifically, for fixed $T$, the \(\eps\shortdash\diffdist\) problem with $\eps=\frac{M}{T}$ consists of $T$ players, who each hold $n$ bits, indexed from $1$ to $n$. 
The players must distinguish between: 
\begin{enumerate}
\item the NO case $\calD^{(n)}_\NO$, in which every bit for every player is drawn i.i.d. from a fair coin and
\item
the YES case $\calD^{(n)}_\YES$, in which an index $L\in[n]$ is selected arbitrarily and the $L$-th bit of each player is chosen i.i.d. from a Bernoulli distribution with parameter $\left(1-\frac{M}{T}\right)$, while all other bits for every player are chosen i.i.d. from a fair coin. 
\end{enumerate}
At a high level, the proof proceeds in two steps:
\begin{enumerate}
    \item First, we prove a communication complexity lower bound for  \(\eps\shortdash\diffdist\) against any protocol that succeeds with probability $1-2^{-\Theta(T)}$, which includes deterministic protocols. 
    \item Second, we show that the  \(\eps\shortdash\diffdist\) problem can be reduced to the expert prediction problem in the streaming setting. 
\end{enumerate}
The second step is straightforward, and the idea was proposed by \cite{SrinivasWXZ22}. 
In the reduction, each player in an instance of \(\eps\shortdash\diffdist\) corresponds to a day of the expert problem. 
The $n$ bit input held by each player correspond to the $n$ expert predictions of each day. 
Therefore, in the NO case, each expert is correct on roughly half of the days. 
In the YES case, there is a single expert $L \in [n]$ that is correct on roughly $1/2+\delta$ of the days (for $\delta= 1/2 - M/T$), while all other experts randomly guess each day.  
Suppose that there is a streaming algorithm for the expert prediction problem  with average regret $\delta/2$. 
Then roughly speaking, in the YES case, the algorithm is correct approximately on $1/2+\delta/2$ of the days, while in the NO case where every expert is randomly guessing, the algorithm is correct on less than $1/2+\delta/2$ of the days. This distinguishes the YES and NO case  and thus solves \(\eps\shortdash\diffdist\).

For the second step, we   show that solving the \(\eps\shortdash\diffdist\)  problem with probability at least $1-2^{-\Theta(T)}$ requires $\Omega(nM)$ total communication. 

Observe that if the input is viewed as a $T\times n$ matrix, then $\calD^{(n)}_\NO$ is a product distribution across columns that can be written as $\zeta^n$, where $\zeta$ is the distribution over a single column such that all entries of the column are i.i.d.\ Bernoulli with parameter $\frac{1}{2}$. 
We view $\calD^{(n)}_\NO$ as a hard distribution and applies an information complexity analysis. 
By a direct sum argument, it suffices to show that the single column problem, i.e., distinguishing between $\calD^{(1)}_\NO$ and $\calD^{(1)}_\YES$ (i.e., for $n=1$), requires $\Omega(M)$ total communication. 


Let $(C_1,C_2,\ldots, C_T)$ be a single column drawn from the hard distribution---namely, the NO case where each player holds  one i.i.d.\ Bernoulli with parameter $1/2$. 
Let $A$  be a fixed protocol  with success probability at least $1-\exp(-\Theta(T))$. For all $i<T$, let $M_i$ denote the message sent from player $P_i$ to player $P_{i+1}$ and $M_{<i} =\{M_j : j<i\}$. 
Let $\Pi= \Pi(C_1,\cdots,C_T)$  be the communication transcript of $A$ given the input $(C_i)_{i=1}^T$. A standard information complexity argument \cite{Bar-YossefJKS04} implies that the total communication   is at least the \textit{information cost}, defined as $I(C_1,\ldots,C_T;\Pi(C_1,\ldots,C_T))$, where $I(X,Y)$ denotes the mutual information between random variables $X$ and $Y$.  

The key step of our proof is therefore to lower bound the information cost by $\Omega(M)$. The main ideas are the following.
For any $i \in [T]$, we say that $(M_i, M_{< i})$ is \emph{informative} for $i$ with respect to the input $C$ and the transcript $\Pi=(M_1,M_2,\ldots,M_T)$ if 
\begin{equation}
   \left| \Pr \left(C_i = 0  \mid M_i, M_{< i} \right) -  \Pr \left(C_i = 1  \mid M_i, M_{< i} \right) \right| \geq c
\end{equation}
for some constant $c > 0$. 
Otherwise, we say that $M_i$ is uninformative so that intuitively, an informative message $M_i$ reveals sufficiently large information about $C_i$ so that the mutual information $I(M_i, C_i \mid M_{< i})$ would be large. 

Now for all $i\in [T]$, let $p_i $ be the probability that $(M_i, M_{< i})$ is informative (for $i$ with respect to $C$ and $\Pi$), taken over all possible inputs and randomness used in the protocol.  
It is straightforward to show that
\[I(\Pi; C_1,C_2,\ldots, C_T) = \sum_{j=1}^T  I\left(M_j ;  C_j \mid  M_{<j}\right) \geq \Omega \left(\sum_{j=1}^T p_j\right).\]
Namely, we first use the chain rule for mutual information to decompose the mutual information into the individual terms in the summation, which can be further decomposed using conditional entropy. 
Then the desired bound immediately follows from a standard bound on the binary entropy function. 
From here, it suffices to prove that
\[\sum_{j=1}^T p_j>\gamma\cdot M,\]
for some fixed constant $\gamma>0$. 

To that end, we observe that if $\sum_{j=1}^T p_j=o(M)$ then by Markov's inequality, the probability that the set $S$ of uninformative indices has size at least $T-o(M)$ is at least $\frac{9}{10}$. 
We show that by modifying $C$ on the uninformative indices $S$, we can find an input $C'$ on which $A$ cannot guarantee correctness with probability at least $1-\exp(-\Theta(T))$. 
Let $C'$ be an input that agrees with $C$ on the informative indices $[T]\setminus S$, so that $C'_i=C_i$ for $i\in[T]\setminus S$, and is chosen arbitrarily on uninformative indices $S$. 
By definition of informative index, the probability that the protocol $A$ generates $\Pi$ on input $C'$ is at least $(1-c)^T\ge e^{-cT}$ times the probability that the protocol $A$ generates $\Pi$ on input $C$. 
However, since $C$ can differ from $C'$ on $S$, then $C$ can differ from $C'$ on $|S|=T-o(M)$ indices and it follows that there exists a choice of $C'$ that contains fewer than $\frac{M}{2}$ zeros such that $A$ will also output $\Pi$ with probability at least $\frac{e^{-cT}}{2}$. 
On the other hand, since $\Pi$ corresponds to a transcript for which $A$ will output NO, then $A$ cannot succeed with probability $1-\frac{e^{-cT}}{8}$ on the input $C'$. 
On the other hand, a YES instance will generate $C'$ with probability $2^{-T}$, which is a contradiction, and thus it follows that $\sum_{j=1}^T p_j=\Omega(M)$, as desired.  

\paragraph{Algorithms for adaptive inputs.}
On the positive side, we show that there exists a randomized algorithm for the discrete prediction with experts problem that is robust to adaptive inputs:
\begin{theorem}[Robust algorithms against adaptive inputs]
\label{thm:experts:dp}
Let $R > \frac{64\log^2 n}{T}$, and suppose the best expert makes at most $M\le\frac{R^2 T}{128\log^2 n}$ mistakes. 
Then there exists an algorithm for the discrete prediction with experts problem that uses $\widetilde{O}\left(\frac{n}{R\sqrt{T}}\right)$ space and achieves regret at most $R$, with probability at least $1-\frac{1}{\poly(n,T)}$.
\end{theorem}
We remark that~\autoref{thm:experts:dp} provides a smooth trade-off between the space and regret, almost all the way to the information-theoretic limit of $R=O_n\left(\sqrt{\frac{1}{T}}\right)$ for general worst-case input. 
However, it incurs a multiplicative space overhead of $\widetilde{O}(\sqrt{T})$ compared to the optimal algorithms for oblivious input. 
Thus we believe the complete characterization of the space complexity of the discrete prediction with experts problem with adaptive input is a natural open question resulting from our work. 

Our algorithm for~\autoref{thm:experts:dp} uses differential privacy (DP) to hide the internal randomness of our algorithm from the adaptive adversary. 
The technique was first proposed by the recent work \cite{HassidimKMMS20, AttiasCSS23,BeimelKMNSS22} to achieve adversarial robustness in data streaming algorithms.
To exploit it for solving our problem of expert learning, we run $\tO{\sqrt{T}}$ copies of the oblivious algorithm of \cite{SrinivasWXZ22} and then use advanced composition to show that running a private median on each of the $\tO{\sqrt{T}}$ copies across $T$ interactions guarantees differential privacy. 
Correctness then follows from generalization of DP. 

\subsection{Related Work}

\paragraph{The experts problem.}
The experts problem has been extensively studied~\cite{CesaBianchi06}, both in the discrete decision setting \cite{LittlestoneW94} and in the setting where costs are determined by various loss functions \cite{HausslerKW95,Vovk90,Vovk98,Vovk99,Vovk05}. 
Hence, the experts problem can be applied to many different applications, such as portfolio optimization \cite{CoverO96, cover1991universal}, ensemble boosting \cite{freund1997decision}, and forecasting \cite{Herrmann22}. 
Given certain assumptions on the expert, such as assuming the experts are decisions trees \cite{HelmboldS97,TakimotoMV01}, threshold functions \cite{MaassW98}, or have nice linear structures \cite{KalaiV05}, additional optimizations have been made to improve the algorithmic runtimes for the experts problem and more generally, existing work has largely ignored optimizing for memory constraints in favor of focusing on time complexity or regret guarantees, thus frequently using $\Omega(n)$ memory to track the performance of each expert.

Recently, \cite{SrinivasWXZ22} introduced the study of memory-regret trade-offs for the experts problem. 
For $n\gg T$, \cite{SrinivasWXZ22} showed that the space complexity of the problem is $\tilde{\Theta}\left(\frac{n}{R^2T}\right)$ in the random-order streams, but also gave a randomized algorithm that uses $\tO{\frac{n}{RT}}$ space for arbitrary-order streams when the number of mistakes $M$ made by the best expert is ``small''. 
Subsequently, \cite{PengZ22} considered the online learning with experts problem for $T\gg n$, introducing a general space-regret trade-off framework that achieves $o(T)$ regret using $o(n)$ memory, including $O_n(T^{4/5})$ regret with $\O{\sqrt{n}}$ space and $O_n(T^{0.67})$ regret with $\O{n^{0.99}}$ space. 

\paragraph{Adaptive inputs.}
Motivated by non-independent inputs and adversarial attacks, adaptive inputs have recently been considered in the centralized model~\cite{CherapanamjeriN20,KontorovichSS22,CherapanamjeriN22,CherapanamjeriSWZZZ23}, in the streaming model~\cite{AvdiukhinMYZ19,Ben-EliezerJWY21,HassidimKMMS20,WoodruffZ21,BravermanHMSSZ21,ChakrabartiGS22,Ben-EliezerEO22,AjtaiBJSSWZ22,AssadiCGS22,Cohen0NSSS22,AttiasCSS23,DinurSWZ23}, and in the dynamic model~\cite{Wajc20,BeimelKMNSS22}. 
In particular, algorithms robust to inputs that can depend on the previous outputs by the algorithm, i.e., black-box attacks, are also robust to situations in which future inputs may be dependent on previous outputs. 
This is especially relevant in applications such as forecasting, in which a prediction on day $i$ can lead to a series of actions that might impact outcomes and expert predictions on day $i+1$ and beyond. 

Adaptive adversaries have received considerable attention in literature for online learning when the goal is simply to achieve the best possible regret~\cite{BorodinE98,CesaBianchi06,McMahanB04}. 
Building off a line of results on multi-armed bandit problems \cite{AuerCFS02,AwerbuchK04,Kleinberg04}, the work of \cite{MerhavOSW02} first considered the experts setting against memory-bounded adaptive adversaries, giving an algorithm with regret $\O{T^{2/3}}$. 
An early paper of \cite{FariasM06} introduced a family of algorithms for adaptive inputs, but provided guarantees using concepts not quite related to the standard definitions of regret. 
More recent works have explored online learning with additional considerations, such as alternative quantities to optimize \cite{DekelTA12}, additional switching costs~\cite{Cesa-BianchiDS13,DekelDKP14,RouyerSC21}, and feedback graphs~\cite{AroraMM19}. 
The closest work to our setting is the recent result by~\cite{PengZ22} showing that no algorithm using space sublinear in $n$ can achieve regret sublinear in $T$ when the input is chosen by an adversary with access to the internal state of the algorithm, i.e., a white-box adversary. 

\paragraph{Concurrent and independent work.} 
Concurrent to our work, \cite{PengR23} considered a variant of the problem where at each time, the algorithm selects an expert instead of a prediction. They then introduce an algorithm robust against an adaptive adversary who observes the specific expert chosen by the algorithm at each time, as well as lower bounds for any algorithm robust to such an adversary. 

One way to ensure adversarial robustness is through deterministic algorithms. On that end, we achieve stronger lower bounds for deterministic algorithms, showing that there must be a dependency on the number $M$ of mistakes made by the best expert, i.e., any deterministic algorithm achieving amortized regret $R$ must use $\widetilde{\Omega}\left(\frac{nM}{RT}\right)$ space. In fact, when the number of mistakes $M$ made by the best expert is sufficiently small, i.e., $M=\O{\frac{R^2T}{\log^2 n}}$ for amortized regret $R$, we give a randomized upper bound that uses \emph{less} space than this lower bound. By comparison, the lower bound of \cite{PengR23} shows that any algorithm achieving $R$ amortized regret must use $\widetilde{\Omega}\left(\sqrt{\frac{n}{R}}\right)$ space, though their lower bound also applies to randomized algorithms. 

Due to the difference in setting, our algorithmic techniques are quite different from those of \cite{PengR23}. We use a recent idea of \cite{HassidimKMMS20, AttiasCSS23,BeimelKMNSS22} to hide the internal randomness of our algorithm from the adversary whereas \cite{PengR23} rotates between groups of experts to prevent an adversary from inducing high regret by making a specific expert bad immediately after it is selected.


\section{Preliminaries}
\paragraph{Notations.}
For any $t\leq n$ and vector $(X_1,X_2,\cdots, X_n)$, we let $X_{<t}$ denote $(X_1,\cdots, X_{t-1})$,  $X_{\le t}= (X_1,\cdots, X_{t})$, and $X_{-t} = (X_1,\cdots, X_{t-1}, X_{t+1},\cdots, X_n)$. 
Also, $X_{>t}$ and $X_{\geq t}$
are defined similarly. Let $e_i$ denote the $i$th standard basis vector, and for any $S$, $e_S$   the vector that has a $1$ at index $i\in S$ and $0$ everywhere else. 
For a random variable $X$, let $H(X)$ denote its entropy.

\subsection{Information Theory} For any $p\in [0,1]$, we slightly abuse notation and  let $H(p) = -p\log_2 p - (1-p) \log_2(1-p) $ be the binary entropy function. The following is a standard  upper and lower bound of $H(p)$.
\begin{lemma}[Bound on the binary entropy function; see e.g. \cite{enwiki:1071507954}]\label{lem:entropy-bound}
For $p\in [0,1]$, the binary entropy function satisfies
\begin{equation*}
4p(1-p)\leq H (p)\leq 2 \sqrt{ p(1-p)}.
\end{equation*}
\end{lemma}

\subsection{Communication Complexity}
\begin{definition}[Mutual information]
Let $X$ and $Y$ be a pair of random variables with joint distribution $p(x,y)$. 
Then the \emph{mutual information} is defined as $I(X;Y):=\sum_{x,y}p(x,y)\log\frac{p(x,y)}{p(x)p(y)}$, for marginal distributions $p(x)$ and $p(y)$.  
\end{definition}

In a multi-party communication problem of $t$ players, each player is given $x_i \in \mathcal{X}_t$. They  communicate according to fixed protocol to compute a function $f : \mathcal{X}_t\times \cdots \times \mathcal{X}_t \rightarrow \mathcal{Y}$. 
A protocol $\Pi$ is called a
$\delta$-error protocol for $f$  if there exists a function $\Pi_{\text{out}}$ such that $\operatorname{Pr}\left[\Pi_{\text {out }}(\Pi(x, y))=f(x, y)\right] \geqslant 1-\delta$.
For a (multi-party) communication problem,   we denote the transcript of all communication in a
protocol as $\Pi \in \{0,1\}^*$. The communication
cost of a protocol, as a result, is the bit length of the transcript. 
Let $R_\delta(f)$ denote the communication
cost of the best $\delta$-error protocol for $f$.
\begin{definition}[Information cost]
Let $\Pi$ be a randomized protocol that produces a random variable $\Pi(X_1,\ldots,X_T)$ as a transcript on inputs $X_1,\ldots,X_T$ drawn from a distribution $\mu$. 
Then the \emph{information cost} of $\Pi$ with respect to $\mu$ is defined as $I(X_1,\ldots,X_T;\Pi(X_1,\ldots,X_T))$.  
\end{definition}

\begin{definition}[Information complexity]
The information complexity of a function $f$ with respect to a distribution $\mu$ and failure probability $\delta$ is the minimum information cost of a protocol for $f$ with respect to $\mu$ that fails with probability at most $\delta$ on every input and denoted by $\IC_{\mu,\delta}(f)$. 
\end{definition}

\begin{lemma}[Information cost decomposition lemma, Lemma 5.1 in~\cite{Bar-YossefJKS04}]
\label{lem:decompose}
Let $\mu$ be a mixture of product distributions and suppose $\Pi$ is a protocol for inputs $(X_1,\ldots,X_T)\sim\mu^n$. 
Then $I(X_1,\ldots,X_T;\Pi(X_1,\ldots,X_T))\ge\sum_{i=1}^n I(X_{1,i},\ldots,X_{T,i};\Pi(X_1,\ldots,X_T))$, where $X_{i,j}$ denotes the $j$-th component of $X_i$. 
\end{lemma}

\begin{lemma}[Information complexity lower bounds communication complexity; Proposition 4.3 \cite{Bar-YossefJKS04}]\label{lem:ic-cc}
    For any distribution $\mu$ and error  $\delta$, $R_\delta (f) \geq \IC_{\mu,\delta}(f)$. 
\end{lemma}

\subsection{Differential Privacy}
Our algorithmic results rely on the following tools from differential privacy.
\begin{definition}[Differential privacy, \cite{DworkMNS06}]
Given a privacy parameter $\eps>0$ and a failure parameter $\delta\in(0,1)$, a randomized algorithm $\calA:\calX^*\to\calY$ is $(\eps,\delta)$-differentially private if, for every pair of neighboring streams $S$ and $S'$ and for all $E\subseteq\calY$,
\[\PPr{\calA(S)\in E}\le e^{\eps}\cdot\PPr{\calA(S')\in E}+\delta.\]
\end{definition}


\begin{theorem}[Private median, e.g.,~\cite{HassidimKMMS20}]
\label{thm:dp:median}
Given a database $\calD\in X^*$, a privacy parameter $\eps>0$ and a failure parameter $\delta\in(0,1)$, there exists an $(\eps,0)$-differentially private algorithm $\privmed$ that outputs an element $x\in X$ such that with probability at least $1-\delta$, there are at least $\frac{|S|}{2}-m$ elements in $S$ that are at least $x$, and at least $\frac{|S|}{2}-m$ elements in $S$ in $S$ that are at most $x$, for $m=\O{\frac{1}{\eps}\log\frac{|X|}{\delta}}$. 
\end{theorem}

\begin{theorem}[Advanced composition, e.g.,~\cite{DworkRV10}]
\label{thm:adaptive:queries}
Let $\eps,\delta'\in(0,1]$ and let $\delta\in[0,1]$. 
Any mechanism that permits $k$ adaptive interactions with mechanisms that preserve $(\eps,\delta)$-differential privacy guarantees $(\eps',k\delta+\delta')$-differential privacy, where $\eps'=\sqrt{2k\ln\frac{1}{\delta'}}\cdot\eps+2k\eps^2$. 
\end{theorem}

\begin{theorem}[Generalization of DP, e.g.,~\cite{DworkFHPRR15,BassilyNSSSU21}]
\label{thm:generalization}
Let $\eps\in(0,1/3)$, $\delta\in(0,\eps/4)$, and $n\ge\frac{1}{\eps^2}\log\frac{2\eps}{\delta}$. 
Suppose $\calA:X^n\to 2^X$ is an $(\eps,\delta)$-differentially private algorithm that curates a database of size $n$ and produces a function $h:X\to\{0,1\}$. 
Suppose $\calD$ is a distribution over $X$ and $S$ is a set of $n$ elements drawn independently and identically distributed from $\calD$. 
Then
\[\PPPr{S\sim\calD,h\gets\calA(S)}{\left|\frac{1}{|S|}\sum_{x\in S}h(x)-\EEx{x\sim\calD}{h(x)}\right|\ge10\eps}<\frac{\delta}{\eps}.\]
\end{theorem}

\section{Lower Bounds for Arbitrary-Order Streams}
In this section, we give space lower bounds for the experts problem on arbitrary-order streams. 
As a warm-up, we first show in~\autoref{sec:lb:acc} a general space lower bound for randomized algorithms when the best expert makes a ``small'' number of mistakes. 
We then give our main lower bound result in~\autoref{sec:lb:det}, showing that any deterministic algorithm achieving regret $R$ must use space $\Omega\left(\frac{nM}{RT}\right)$ when the best expert makes $M$ mistakes. 

\subsection{Warm-up: Lower Bound for Accurate Best Expert}
\label{sec:lb:acc}
In this section, we show that any randomized algorithm that achieves regret $R$ must use $\Omega\left(\frac{n}{RT}\right)$ space, even when the best expert makes $\Theta(RT)$ mistakes. 
In contrast, \cite{SrinivasWXZ22} give an $\Omega\left(\frac{n}{R^2T}\right)$ space lower bound:

\begin{theorem}[Memory lower bound; Theorem 1 of \cite{SrinivasWXZ22}]\label{thm:LowerBound}
Let $R>0$, $p<\frac{1}{2}$ be fixed constants, i.e., independent of other input parameters.  
Any algorithm that achieves $R$ regret for the experts problem with probability at least $1-p$ must use at least $\Omega\left(\frac{n}{R^2T}\right)$ space.

Furthermore, this lower bound holds even when the costs are binary, and expert predictions, as well as the correct answers, are constrained to be i.i.d.\ across the days, albeit with different distributions across the experts.
\end{theorem}

The proof of this lower bound exploits a construction where the best expert makes $\Theta(T)$ mistakes. 
Thus, it is not clear how the space complexity of the problem behaves when the best expert makes a smaller number of mistakes. 
In fact, \cite{SrinivasWXZ22} also give an {algorithm} that uses $\widetilde{O}\left(\frac{n}{RT}\right)$ space when the best expert makes $O(RT)$ mistakes, bypassing the aforementioned lower bound. 

We now prove that in this small mistake regime, this algorithm is tight. Towards this goal, we first define the \(\eps\shortdash\diffdist\) problem that   reduces to the experts problem. It was proposed by \cite{SrinivasWXZ22} to prove  memory lower bounds for the expert problem in random order stream. 
\begin{definition}[The \(\eps\shortdash\diffdist\) Problem]
\label{def:diff-dist}
We have $T$ players, each of whom holds $n$ bits, indexed from $1$ to $n$.  We must distinguish between two cases, which we refer to as ``$V = 0$" and ``$V = 1$". Let \(\mu_0\) be a Bernoulli distribution with parameter \(\frac{1}{2}\), i.e., a fair coin, and let \(\mu_1\) be a Bernoulli distribution with parameter \(\frac{1}{2} + \eps\).
\begin{itemize}
\item (NO Case, ``$V = 0$") Every index for every player is drawn i.i.d. from a fair coin, i.e.,\ \(\mu_0\).
\item (YES Case, ``$V = 1$") An index $L\in[n]$ is selected arbitrarily---the $L$-th bit of each player is chosen i.i.d.\ from $\mu_1$. All other bits for every player are chosen i.i.d.\ from $\mu_0$.
\end{itemize}
\end{definition}
Any protocol that successfully solves the \(\eps\shortdash\diffdist\) problem with a constant probability greater than $\frac{1}{2}$ must use at least $\Omega\left(\frac{n}{\eps^2}\right)$ communication, a result due to \cite{SrinivasWXZ22}:
\begin{lemma}[Communication complexity of \(\eps\shortdash\diffdist\); Lemma 3 of \cite{SrinivasWXZ22}]\label{lemma:InfoOfDiffDist}
The communication complexity of solving the \(\eps\shortdash\diffdist\) problem with a constant \(1 - p\) probability, for any \(p \in [0, 0.5)\), is $\Omega\left(\frac{n}{\eps^2}\right)$.
\end{lemma}

The proof of~\autoref{thm:LowerBound} by \cite{SrinivasWXZ22} uses $n$ coin flips across each of the $T$ players to form the $n$ expert predictions over each of the $T$ days. 
In the NO case, each   expert will be correct on roughly $\frac{T}{2}$ days, while in the YES case, a single expert will be correct on roughly $\frac{T}{2}+\eps T$ days, so that an algorithm with regret $R=O(\eps)$ will be able to distinguish between the two cases. 
There is a slight subtlety in the proof that uses a masking argument to avoid ``trivial'' algorithms that happen to succeed on a ``lucky'' input, but for the purposes of our proof in this section, the masking argument is not needed. 
It then follows that the total communication is $\Omega\left(\frac{n}{R^2}\right)$ across the $T$ players, so that any streaming algorithm must use at least $\Omega\left(\frac{n}{R^2T}\right)$ bits of space. 

Suppose we instead consider the \(\eps\shortdash\diffdist\) problem over $RT$ players, representing $RT$ days in the experts problem. 
Moreover, suppose we set $\eps=\Theta(1)$ in the \(\eps\shortdash\diffdist\) problem, so that in the NO case, each of the experts will be correct on roughly $\frac{RT}{2}$ days, while in the YES case, a single expert will be correct on roughly $\frac{RT}{2}+CRT$ days, for some constant $C>0$. 
Suppose we further pad all of the experts with incorrect predictions across an additional $T-RT$ days, so that the total number of days is $T$, but the number of correct expert predictions remains the same. 
Then an algorithm achieving regret $O(R)$ will be able to distinguish between the two cases, so that the total communication is $\Omega\left(\frac{n}{R}\right)$, so that any streaming algorithm must use at least $\Omega\left(\frac{n}{RT}\right)$ bits of space. 

\begin{corollary}
Let $R$, $p<\frac{1}{2}$ be fixed constants, i.e., independent of other input parameters.  
Any algorithm that achieves $R$ regret for the experts problem with probability at least $1-p$ must use at least $\Omega\left(\frac{n}{RT}\right)$ space even when the best expert makes as few as $\Theta(RT)$ mistakes. 
This lower bound holds even when the costs are binary and expert predictions, as well as the correct answer, are constrained to be i.i.d.\ across the days, albeit with different distributions across the experts.
\end{corollary}
\begin{proof}
The claim follows from setting $T=RT$ and $R=\Theta(1)$ in the proof of~\autoref{thm:LowerBound}. 
\end{proof}

\subsection{Lower Bound for Deterministic Algorithms}
\label{sec:lb:det}
We now prove our main space lower bound for deterministic algorithms (\autoref{thm:main-intro}). 
We first set up some basic notations and introduce a hard distribution.

Let $T$ be any fixed positive integer. Let $\calD^{(n)}_\NO$ be the distribution over matrices $A$ with size $T\times n$ such that all entries of the matrix are i.i.d.\ Bernoulli with parameter $\frac{1}{2}$, i.e., each entry of $A$ is $0$ with probability $\frac{1}{2}$ and $1$ with probability $\frac{1}{2}$. 
Let $\calD^{(n)}_\YES$ be the distribution over matrices $M$ with size $T\times n$ such that there is a randomly chosen column $L\in[n]$, which is i.i.d.\ Bernoulli with parameter $\left(1-\frac{M}{T}\right)$ and all other columns are i.i.d.\ Bernoulli with parameter $\frac{1}{2}$. 
Let $\biasdetect_n$ be the problem of detecting whether $A$ is drawn from $\calD^{(n)}_\YES$ or $\calD^{(n)}_\NO$. 

Let $\Pi$ be a communication protocol for $\biasdetect_n$  that is correct with probability at least $1-\exp(-\Theta(T))$.
Since $\calD^{(n)}_\NO$ is a product distribution across columns, then it can be written as $\zeta^n$, where $\zeta$ is the distribution over a single column such that all entries of the column are i.i.d.\ Bernoulli with parameter $\frac{1}{2}$. 
Let  $\biasdetect_1$ denote the problem  of distinguishing between $\calD^{(1)}_\NO$ and $\calD^{(1)}_\YES$ on a single column, i.e., $n=1$. Using  $\calD^{(n)}_\NO$ as the hard distribution,  we have the following direct sum theorem.
\begin{lemma}[Direct sum for \textsc{BiasDetect}]\label{lem:direct-sum}
The information complexity of $\biasdetect_n$  satisfies 
\[\IC_{\calD_\NO^{(n)},2^{-\Theta(T)}}(\biasdetect_n)\ge n\cdot \IC_{\calD_\NO^{(1)},2^{-\Theta(T)}}(\biasdetect_1).\]
\end{lemma}
\begin{proof}
    By definition, $\calD_\NO^{(n)}  = \zeta^n$ is a product distribution over $n$ columns. The lemma follows from   the standard direct sum lemma of information cost (\autoref{lem:decompose}).
\end{proof}



With the above direct sum theorem for $\biasdetect_n$, it now suffices to provide a single-coordinate information cost lower bound against $\biasdetect_1$.  The proof is delayed to  Section \ref{sec:key-lemma-pf}.
 \begin{lemma}[Single-coordinate information cost lower bound]\label{lem:single-lb}
Let $c\in (0,1)$ and $\Pi$ be any protocol with  error $\delta = 2^{-\Theta(T)}$ for $\biasdetect_1$.  We have that the information cost of $\Pi$ with respect to $\zeta$ is at least 
\begin{equation}\label{eqn:single-coord-ic}
    I (\Pi (C_1,C_2,\ldots , C_T); C_1,C_2,\ldots, C_T) \geq \Omega \left(  M\right),
\end{equation}
where the bits $C_i \sim \zeta$ are i.i.d.\ single coordinates.
\end{lemma}
Combining \autoref{lem:single-lb} with the direct sum theorem (\autoref{lem:direct-sum}), we immediately 
get the following information complexity lower bound for $\biasdetect_n$:
\begin{theorem}[$n$-Coordinate information complexity lower bound]
\label{thm:n:lb}
Let $c \in (0,1)$. Then $$\IC_{\calD_\NO^{(n)},2^{-\Theta(T)}}(\biasdetect_n)=\Omega (  nM) . $$
\end{theorem}
\begin{proof}
    This follows by applying the direct sum theorem  (\autoref{lem:direct-sum}) to the single-coordinate bound   \autoref{lem:single-lb}.
\end{proof}
This implies that  any  algorithm  with $R$ regret and  success rate at least $1-2^{-\Theta(T)}$ requires $\Omega \left(\frac{nM}{RT}\right)$ memory, where $M$ is the mistake bound on the best expert.

\begin{theorem}[Memory lower bound for expert learning]
\label{thm:lb:random}
Let $R,M$ be fixed and independent of other input parameters. 
Any streaming algorithm that achieves $R$ regret for the experts problem with probability at least $1-2^{-\Theta(T)}$ must use at least $\Omega(\frac{nM}{RT})$ space, for $n=o\left(2^T\right)$, where the best expert makes $M$ mistakes.\end{theorem}
\begin{proof}
We now consider the problem $\biasdetect_n$ on a matrix of size $RT\times n$. 
Note that in the NO case,  at  any fixed column $i\in[n]$, the probability that there are more than $\frac{3RT}{5}-\frac{M}{2}$ instances of $0$, for $M\le\frac{RT}{8}$, is at most $2\exp(-c_1RT)$, for a sufficiently small constant $c_1\in(0,1)$.  
Thus, by a union bound, the probability that there exists an index $i\in[n]$ with more than $\frac{3RT}{4}-\frac{M}{2}$ instances of $0$ is at most $2n\exp(-c_1RT)$. 

Similarly in the YES case, the probability that there are fewer than $\frac{4RT}{5}-\frac{M}{2}$ instances of $0$ for a fixed $i\in[n]$ and for $M\le\frac{RT}{8}$ is at most $2\exp(-c_2RT)$, for a sufficiently small constant $c_2\in(0,1)$ and so by a union bound, the probability that there exists an index $i\in[n]$ with fewer than $\frac{3RT}{4}-\frac{M}{2}$ instances of $0$ is at most $2n\exp(-c_2RT)$. 
Hence, for $n=o(2^T)$, there exists a constant $c\in(0,1)$ such that any algorithm that achieves total regret at most $\frac{RT}{5}$ with probability at least $1-\exp(-cT)$ can distinguish between the YES and NO cases with probability $1-\exp(-\Theta(T))$. 

By~\autoref{thm:n:lb} and \autoref{lem:ic-cc}, the total communication across the $RT$ players must be at least $\Omega(nM)$. 
Therefore, any streaming algorithm that achieves average $R$ regret for the experts problem with probability at least $1-2^{-\Theta(T)}$ must use at least $\Omega(\frac{nM}{RT})$ space.
\end{proof}

\subsection{Proof of the Single-Coordinate Information Cost Lower Bound}\label{sec:key-lemma-pf}
We now show the single-coordinate lower bound of \autoref{lem:single-lb}.
\begin{proof}[Proof of \autoref{lem:single-lb}]
Consider a protocol that is correct with probability $1-2^{-\Theta(T)}$ and let $(C_1,C_2,\ldots , C_T) \sim \zeta^T$ be a single column drawn from the NO case, where each coordinate is i.i.d.\ Bernoulli with parameter $1/2$. For notational convenience, let $\Pi = \Pi (C_1,\cdots , C_T)$ denote the transcript given the input $(C_1,C_2,\cdots , C_T)$.   We consider the one-way message-passing model, where each player $P_i$ holds the input $C_i$. For all $i<T$, let $M_i$ denote the message sent from player $P_i$ to player $P_{i+1}$.   

By the chain rule of mutual information, the information cost of the transcript,  the left-side of \autoref{eqn:single-coord-ic} that we need to bound, can be written as 
\begin{align}
     I(\Pi; C_1,C_2,\ldots, C_T) =  \sum_{j= 1}^{T} I\left(M_j ;  C_1,C_2,\ldots, C_T \mid  M_{<j}\right).
\end{align}
By the independence of one-way communication,  we have
\begin{align}
    I\left(M_j ;  C_1,C_2,\ldots, C_T \mid  M_{<j}\right) = I\left(M_j ;  C_j \mid  M_{<j}\right).
\end{align}
Combining the two equalities above, the information cost equals
\begin{equation}\label{eqn:ic-24}
     I(\Pi; C_1,C_2,\ldots, C_T) = \sum_{j=1}^T  I\left(M_j ;  C_j \mid  M_{<j}\right).
\end{equation}
We  now lower bound the right-side. First, we make the following definition.
For any $i \in [T]$, we say that $(M_i, M_{< i})$ is \textit{informative} for $i$ with respect to the input $C$ and the transcript $\Pi=(M_1,\ldots,M_T)$ if 
\begin{equation}
   \left| \Pr (C_i = 0  \mid M_i, M_{< i} ) -  \Pr (C_i = 1  \mid M_i, M_{< i} ) \right| \geq c
\end{equation}
for some constant $c > 0$; and uninformative otherwise. 
Intuitively, an informative index $i$ with respect to $(M_i, M_{< i})$  means that conditional on the past messages $ M_{< i}$,  the message $M_i$ reveals much information about $C_i$.  
Hence, in this case, $I(M_i, C_i \mid M_{< i})$ would be large. 
Now for all $i\in [T]$, let $p_i $ be the probability that $(M_i, M_{< i})$ is  informative (for $i$ with respect to $C$ and $\Pi$). 

Conceptually, we   need to show that $\sum_i p_i$ is large, since then there would be sufficiently many informative  messages, and so the information cost in the left-side of \autoref{eqn:ic-24} is high. We formalize this idea in the following lemma.
\begin{lemma}\label{lem:ic-sum-pj}
In the setting above, where $c >0$ is a constant, the information cost can be lower bounded by 
\begin{align}\label{eqn:sum-ic-26}
  I(\Pi; C_1,C_2,\ldots, C_T) = \sum_{j=1}^T  I\left(M_j ;  C_j \mid  M_{<j}\right) \geq \Omega \left(\sum_{j=1}^T p_j\right)
\end{align}
\end{lemma}
\begin{proof}
We start by  expanding the definition of the mutual information terms. For each  $j \in T$, we have 
\begin{align}\label{eqn:def-ci-27}
     I\left(M_j ;  C_j \mid  M_{<j}\right) &= H\left(C_j \mid M_{<j} \right) - H\left(C_j \mid M_j, M_{<j}\right)
\end{align}
For the first term, notice that $C_j$ and $M_{<j}$ are independent by one-way communication. Moreover, by definition $C_j$ is Bernoulli with parameter $1/2$. Therefore, \[H(C_j \mid M_{<j}) = H(C_j) = H(1/2) = 1.\]  
For the second term, 
\begin{itemize}
    \item either $(M_j, M_{<j})$ is informative, which holds with probability $p_j$, and in this case, the conditional entropy is upper bounded by $ H\left(C_j \mid M_j, M_{<j}\right) \leq H(1/2 + c/2)$; 
    \item or $(M_j, M_{<j})$ is uninformative, and in this case, we trivially upper bound the conditional entropy by $ H\left(C_j \mid M_j, M_{<j}\right) \leq 1$; 
\end{itemize} 
Putting the observations together and using  \autoref{eqn:def-ci-27}, it follows   that 
\begin{align*}
    I\left(M_j ;  C_j \mid  M_{<j}\right) &=H\left(C_j \mid M_{<j} \right) - H\left(C_j \mid M_j, M_{<j}\right)\\
    &\ge 1 - (p_j \cdot H(1/2 + c/2) + (1-p_j) \cdot 1) \\
    &= p_j - p_j \cdot H(1/2 + c/2)\\
    &\ge p_j  \left(1- \sqrt{1-c^2}\right)\\
    &\ge c^2\cdot  \Omega(p_j),
\end{align*}
where the second last step uses the upper bound of \autoref{lem:entropy-bound} and the last step follows since $1-\sqrt{1-x^2} \geq x^2/5$ for $x\in [0,1]$.
Summing over $j = 1,2,\ldots, T$ in \autoref{eqn:sum-ic-26} finishes the proof.
\end{proof}
To prove the claimed information cost inequality \autoref{eqn:single-coord-ic}, we show that $\sum_i p_i = \Omega(M)$. 

\begin{lemma}\label{lem:sum-pj}
There exists a constant $\gamma>0$ such that 
\[\sum_{j=1}^T p_j>\gamma\cdot M.\]
\end{lemma}
\begin{proof}
Suppose by way of contradiction that $\sum_{j=1}^T p_j=o(M)$. 
Let $A$ be a protocol that sends (possibly random) messages $M_1,\ldots,M_T$ on a random input $C\in\{0,1\}^T\sim\zeta^T$ drawn uniformly from the NO distribution, i.e., each coordinate of $C:=C_1,\ldots,C_T$ is picked to be $0$ with probability $\frac{1}{2}$ and $1$ with probability $\frac{1}{2}$. 
Moreover, suppose $A$ is a protocol that distinguishes between a YES instance and a NO instance with probability at least $1-\frac{e^{-cT}2^{-T}}{8}$, for some constant $c>0$. 

Since $p_i$ is the probability that $M_i$ is informative, then by assumption, the expected number of informative indices $i$ over the messages $M_1,\ldots,M_T$ is $f(M)$ for some $f(M)=o(M)$. 
Thus by Markov's inequality, the probability that the number of informative indices is at most $10f(M)=o(M)$ with probability at least $\frac{9}{10}$. 
Let $S$ be the set of the uninformative indices so that $|S|=T-10f(M)=T-o(M)$. 
Let $C'$ be an input that agrees with $C$ on the informative indices $[T]\setminus S$ and is chosen arbitrarily on uninformative indices $S$, so that $C'_i=C_i$ for $i\in[T]\setminus S$. 
 
By definition, each uninformative index only changes the distribution of the output by a $(1\pm c)$ factor. 
In particular, for $c\in(0,1/2)$, the probability that the protocol $A$ generates $\Pi$ on input $C'$ is at least $(1-c)^T\ge e^{-2cT}$ times the probability that the protocol $A$ generates $\Pi$ on input $C$. 
However, since $C$ can differ from $C'$ on $S$, then $C$ can differ from $C'$ on $|S|=T-10f(M)=T-o(M)$ indices. 

Now since each coordinate of $C$ is picked to be $0$ with probability $\frac{1}{2}$ and $1$ with probability $\frac{1}{2}$, then the probability that $C$ contains more than $T-M$ zeros is at least $1-T^M\cdot\frac{1}{2^T}\ge 1-2^{T/2}$ for sufficiently large $T$. 
But then there exists a choice of $C'$ that contains fewer than $\frac{M}{2}$ zeros such that $A$ will also output $\Pi$ with probability at least $\frac{e^{-cT}}{2}$. 
Since $C'$ contains fewer than $\frac{M}{2}$, then $C'$ is more likely to generated from a YES instance and indeed a YES instance will generate $C$ with probability $2^{-T}$. 
On the other hand, since $\Pi$ corresponds to a transcript for which $A$ will output NO, then the probability that $A$ is incorrect on $C'$ is at least $\frac{e^{-cT}}{4}$, which contradicts the claim that $A$ succeeds with probability $1-\frac{e^{-cT}2^{-T}}{8}$. 
Thus it follows that $\sum_{j=1}^T p_j=\Omega(M)$, as desired.  
\end{proof}

Now we combine \autoref{lem:ic-sum-pj} and \autoref{lem:sum-pj}. This implies that the information cost can be lower bounded by
\begin{equation}
     I(\Pi; C_1,C_2,\ldots, C_T) \geq  \Omega \left(\sum_{j=1}^T p_j\right) \geq \gamma M,
\end{equation}
for a constant $\gamma >0$. This completes the proof.
\end{proof}


\section{Algorithms Against Adaptive Adversaries}
In this section, we show that there exists  algorithms for the discrete prediction with experts problem that is robust to adaptive outputs. 

\subsection{A Near-Optimal Deterministic Algorithm}\label{sec:det-alg}
We first present a simple deterministic algorithm for arbitrary-order streams with oblivious inputs.

\begin{algorithm}[!htb]
\caption{Deterministic algorithm for the experts problem}
\alglab{alg:det}
\begin{algorithmic}[1]
\Require{A stream of length $T$ with $n$ experts, upper bound $M$ on the number of mistakes made by the best expert, and target regret $R$}
\Ensure{A sequence of predictions with regret $R$}
\State{$k\gets\O{\frac{nM}{RT}\log n}$}
\State{$S \gets \emptyset$}
\While{the stream persists}
\If{$S$ is empty} \Comment{We have cycled through all $n$ experts once}
\State{$S \leftarrow [n]$}
\EndIf
\State{Let $P$ be the first $k$ indices of $S$}
\State{$S\gets S\setminus P$}
\While{$P\neq\emptyset$}
\State{For each following day, choose the outcome output by the majority of the experts in $P$}
\State{Delete the incorrect experts on that day}
\EndWhile
\EndWhile
\end{algorithmic}
\end{algorithm}

We now justify the correctness and space complexity of \algref{alg:det}. 
\begin{theorem}
\label{thm:det}
Among $n$ experts in a stream of length $T$, suppose the best expert makes $M$ mistakes. 
There exists a deterministic algorithm that uses space $\widetilde{O}\left(\frac{nM}{RT}\right)$ and achieves regret $R$.
\end{theorem}
\begin{proof}
We first remark that the algorithm can make at most $\log k\le\log n$ mistakes over the lifespan of each pool of size $k:=\frac{2nM}{RT}\log n$ because each time the algorithm makes a mistake, at least half of the pool must be incorrect and deleted, so the size of the pool decreases by at least half with each mistake the algorithm mistakes. 

Since each pool $P$ has size $k$ and there are $n$ experts, then there are at most $\frac{2n}{k}$ pools before the entire set $S$, which is initialized to $n$, is depleted. 
Thus, there are at most $\frac{2n}{k}$ pools to iterate through the entire set of experts. 
Moreover, each time the algorithm has iterated through the entire set of experts, each expert must have made at least one mistake. 
This is because an expert is only deleted from the pool $P$ when it has made a mistake and since all experts have been deleted from $P$, then all experts have made at least one mistake. 

Since the best expert makes at most $M$ mistakes, then the best expert can be deleted from the pool $P$ at most $M$ times. 
In other words, the algorithm can cycle through the entire set of $n$ experts at most $M$ times. 

Hence, the total number of mistakes by the algorithm is at most \[\frac{2n}{k}\cdot\log n\cdot M\le\frac{2nRT}{2nM\log n}\cdot\log n\cdot M=RT,\]
so the algorithm achieves regret at most $R$. 
Since the algorithm selects a subset of  $k=\frac{2nM}{RT}\log n$ experts, then the space complexity follows. 
\end{proof}

In light of~\autoref{thm:lb:random}, it is evident that~\autoref{thm:det} is nearly optimal, up to polylogarithmic factors, for deterministic algorithms, which are automatically adversarially robust. 
On the other hand, it does not seem necessary that any adversarially robust algorithm must be deterministic. 
Indeed, we now give a randomized adversarially robust algorithm with better space guarantees. 

\subsection{A Randomized Robust Streaming Algorithm}

We first recall the following randomized algorithm for arbitrary-order streams with oblivious input, i.e., non-adaptive input:

\begin{lemma}[Algorithm for oblivious inputs; \cite{SrinivasWXZ22}]\label{lem:disc:pred}
Let $R > \frac{16 \log^2 n }{T}$, and suppose the best expert makes at most $M\le\frac{\delta T}{128\log^2 n}$ mistakes. 
Then there exists an algorithm $\discpred$ for the discrete prediction with experts problem that uses $\widetilde{O}\left(\frac{n}{RT}\right)$ space and achieves regret at most $R$, with probability at least $1-\frac{1}{\poly(n,T)}$. 
\end{lemma}

The algorithm of~\autoref{lem:disc:pred} proceeds by sampling pools of $k=\widetilde{O}\left(\frac{n}{RT}\right)$ experts and running majority vote on the pool, while iteratively deleting poorly performing experts until no experts remain in the pool, at which a new pool of $k$ experts is randomly sampled. 
The main intuition is that either the pool of experts will perform well and achieve low regret, or the pool will be continuously re-sampled until the best expert is sampled multiple times, after which point it will not be deleted from the pool. 
Unfortunately, it is not evident that this algorithm is robust to adaptive inputs because an adversary can potentially learn the experts in each sampled pool and force the experts to make mistakes only on days in which they are sampled by the algorithm. 

Instead, we use differential privacy to hide the internal randomness of the algorithm and in particular, the identity of the experts that are sampled by each pool. 
We first run $\widetilde{O}(\sqrt{T})$ copies of the algorithm and then output the private median of the $\widetilde{O}(\sqrt{T})$ copies, guaranteeing roughly $\left(\frac{1}{\widetilde{O}(\sqrt{T})},0\right)$-differential privacy because we use $\widetilde{O}(\sqrt{T})$ copies of the algorithm. 
Advanced composition, i.e.,~\autoref{thm:adaptive:queries}, then ensures $\left(O(1),\frac{1}{\poly(n)}\right)$-differential privacy, so that correctness then follows from the generalization properties of DP, i.e.,~\autoref{thm:generalization}.  
We give our algorithm in full in~\algref{alg:dp:alg}.

\begin{algorithm}
\caption{Randomized, robust streaming algorithm for the experts problem}
\alglab{alg:dp:alg}
{\textbf{Input: } A stream of length $T$ with $n$ experts and a target regret $R$}\;

{\textbf{Output: } A sequence of predictions with regret $R$}\;
\begin{algorithmic}[1]

\State{Run $m=\O{\sqrt{T}\log(nT)}$ independent instances of $\discpred$ with regret $\frac{R}{4}$}

\State{Run $\privmed$ on the $m$ instances with privacy parameter $\eps=\O{\frac{1}{\sqrt{T}\log(nT)}}$ and failure probability $\delta=\frac{1}{\poly(n,T)}$}

\State{At each time $t\in[T]$, select the output of $\privmed$}
\end{algorithmic}
\end{algorithm}

We now show the correctness of our algorithm on adaptive inputs. 
\begin{theorem}[Algorithm for adaptive inputs]
Let $R > \frac{64 \log^2 n }{T}$, and suppose the best expert makes at most $M\le\frac{R^2 T}{128\log^2 n}$ mistakes. 
Then there exists an algorithm for the discrete prediction with experts problem that uses $\widetilde{O}\left(\frac{n}{R\sqrt{T}}\right)$ space and achieves regret at most $R$, with probability at least $1-\frac{1}{\poly(n,T)}$.
\end{theorem}
\begin{proof}
Suppose we run $m=\O{\sqrt{T}\log(nT)}$ independent instances of $\discpred$ with regret $\frac{R}{4}$. 
Note that for $R > \frac{64\log^2 n}{T}$, we have $\frac{R}{4}> \frac{16\log^2 n }{T}$, which is a valid input to $\discpred$ in~\autoref{lem:disc:pred}. 
By~\autoref{lem:disc:pred}, each instance succeeds on an arbitrary-order stream with probability at least $1-\frac{1}{\poly(n,T)}$. 
By a union bound over the $m$ instances, all instances succeed with probability at least $1-\frac{1}{\poly(n,T)}$. 
In particular, each instance has regret at most $\frac{R}{4}$, so that the total number of mistakes by each instance is at most $M+\frac{RT}{4}$. 
Thus, the total number of mistakes by all instances is at most $m\left(M+\frac{RT}{4}\right)$. 

To consider an adaptive stream, observe that $\privmed$ is called with privacy parameter $\O{\frac{1}{\sqrt{T}\log(nT)}}$ and failure probability $\frac{1}{\poly(n,T)}$. 
By~\autoref{thm:adaptive:queries}, the mechanism permits $T$ adaptive interactions and guarantees privacy $\O{1}$ with failure probability $\frac{1}{\poly(n,T)}$. 
By~\autoref{thm:generalization}, we have that with high probability, if the output of the algorithm is incorrect, then at least $\frac{m}{3}$ of the instances $\discpred$ are also incorrect. 
Since the total number of mistakes by all instances is at most $m\left(M+\frac{RT}{4}\right)$, then the total number of mistakes by the algorithm is at most $3\left(M+\frac{RT}{4}\right)\le M+RT$, since $M\le\frac{R^2 T}{128\log^2 n}$. 
Hence, the algorithm achieves $R$ regret with high probability. 

By~\autoref{lem:disc:pred}, each instance of $\discpred$ uses $\widetilde{O}\left(\frac{n}{RT}\right)$ space. 
Since we use $m=\O{\sqrt{T}\log(nT)}$ independent instances of $\discpred$, then the total space is $\widetilde{O}\left(\frac{n}{R\sqrt{T}}\right)$. 
\end{proof}

\section*{Acknowledgements}
We thank Binghui Peng for helpful discussions. 
David P. Woodruff and Samson Zhou were supported by a Simons Investigator Award and by the National Science Foundation under Grant No. CCF-1815840. 
Fred Zhang was supported by ONR grant N00014-18-1-2562.

\def\shortbib{0}
\bibliographystyle{alpha}
\bibliography{references}

\newpage
\appendix
\section{An Alternative Proof in the Large Mistake Regime}\label{sec:alternative}
We give another analysis of the information cost when $M= \Omega(T)$, where $M$ is the number of mistakes of the best expert. 

 \begin{lemma}[Single-Coordinate Information Cost Lower Bound]\label{lem:single-lb2}
Let $c\in (0,1)$ and $\Pi$ be any protocol with  error $\delta = 2^{-T}$ for $\biasdetect_1$. Suppose that the best expert makes $M = c'T$ mistakes for some constant $c'$.  We have that the information cost of $\Pi$ with respect to $\zeta$ is at least 
\begin{equation}
    I (\Pi (C_1,\cdots , C_T); C_1,\cdots, C_T) \geq \Omega \left((1-c)^2 T\right),
\end{equation}
where $C_i \sim \zeta$ are i.i.d.\ single coordinates.
\end{lemma}
Applying direct sum theorem (\autoref{lem:direct-sum}), we  get the following information complexity lower bound for $\biasdetect_n$:
\begin{theorem}[$n$-Coordinate Information Complexity Lower Bound]
Let $c \in (0,1)$ and assume $M= c'T$ for some constant $c'$. Then $$\IC_{\calD^{(1)},2^{-\Theta(T)}}(\biasdetect_n)=\Omega \left((1-c)^2 nT\right) . $$
\end{theorem}

By an argument similar to \autoref{thm:lb:random}, we have:
\begin{theorem}[Memory lower bound for expert learning]
\label{thm:lb:random2}
Let $M = c'T$ for some constant $c'$. 
Any streaming algorithm that achieves constant regret for the experts problem with probability at least $1-2^{-\Theta(T)}$ must use at least $\Omega(n)$ space, where the best expert makes $M$ mistakes.
\end{theorem}
 
For the purpose of  proving \autoref{lem:single-lb2}, we need some technical lemmas. 

\begin{lemma} [Lemma 3.5 of \cite{kamath2021simple}] \label{lem:not-seen}
Consider any   communication protocol $\Pi$ where each player receives
one bit and condition on any fixed input $b \in \{0,1\}^T$. Each player $i$ can be implemented such that, if the other players receive input $b_{-i}$, player $i$
only observes their input with probability $d_\TV (\Pi_b, \Pi_{b \oplus e_i})$.
\end{lemma}

\begin{lemma}[Lemma 3.6 of \cite{kamath2021simple}]\label{lem:large-set}
Let $c\in(0,1)$, $p \in (0,\frac{1-c}{2})$ and $\gamma_c = \frac{1}{c \log (e/c)}$. For a set of binary random variables $Y_1,Y_2,\cdots, Y_k$ such that $\E {\sum_i Y_i} = pk$, there exists a set $S \subset [n]$ of size $ck$ such that $\Pr (Y_j = 0,  \forall j\in S) > e^{{-k/\gamma_c-1}}$.
\end{lemma}

\begin{proof}[Proof of \autoref{lem:single-lb2}]

Let $(C_1,C_2,\cdots , C_n) \sim \zeta^n$ be a single column drawn from the NO case, where each coordinate is i.i.d.\ Bernoulli with parameter $1/2$. Let $M = c'T$ for some constant $c'$. We consider the one-way message-passing model, where for all $i<T$, $M_i$ denotes the message sent from player $P_i$ to player $P_{i+1}$.   It suffices to lower bound
\begin{align*}
     I(\Pi; C_1,\cdots, C_T) =  \sum_{j= 1}^{T} I(\Pi ; C_j | C_{<j}).
\end{align*}
by the chain rule of mutual information.
We claim that for any $j$
\begin{align*}
     I(\Pi ; C_j | C_{<j}) =   I(\Pi ; C_j | C_{-j}). 
\end{align*}
First, by data processing and the one-way nature  of the protocol
\begin{equation*}
I(\Pi ; C_j | C_{<j}) =  I(M_{\leq j}; C_j | C_{<j}).
\end{equation*} 
for any $j$. Now we just need to show that 
\begin{equation*}
  I(M_{\leq j}; C_j | C_{<j}) = I (\Pi ; C_j | C_{-j}).
\end{equation*} 
By chain rule of mutual information, we can write the right-hand side as 
\begin{align*}
    I (\Pi ; C_j | C_{-j}) &= I(M_{\le j }; C_j | C_{-j}) + I(M_{> j} ; C_j | M_{\leq j}, C_{-j}) \\
    &= I(M_{\le j }; C_j | C_{-j}) + I(M_{> j} ; C_j | M_{\leq j}, C_{>j})
\end{align*}
Observe that $M_{>j}$ and $C_j$ are independent, conditional on $M_{\leq j}$ and $ C_{>j}$. Hence, 
\[I(M_{> j} ; C_j | M_{\leq j}, C_{>j}) = 0\] and this proves the claim. 

Let $\Pi_b$ be the distribution of the protocol transcript when the input is fixed to be $b \in \{0,1\}^n$ and $\oplus$ denote the binary XOR. 
Now  we can bound 
\begin{align}
     I(\Pi; C_1,\cdots, C_T) &=  \sum_{j= 1}^{T} I(\Pi ; C_j | C_{<j}) \nonumber\\
     &=   \sum_{j= 1}^{T} I(\Pi ; C_j | C_{-j}) \nonumber\\
     &\ge \frac{1}{8} \frac{1}{2^{T}}  \sum_{b \in \{0,1\}^{T}} \sum_{j=1}^T d^2_{\TV} (\Pi_{b \oplus e_j}, \Pi_b) \nonumber\\
     & \geq \frac{1}{8}\frac{1}{2^{T}}  \sum_{b \in \{0,1\}^{T}} \sum_{j : b_j = 0} d^2_{\TV} (\Pi_{b \oplus e_j}, \Pi_b).\label{eqn:final-line}
\end{align}
Conditioned on an input $b \in \{0,1\}^T$, let $k = | \{ i : b_i = 0\}|$ and assume for the sake of a contradiction that 
\begin{equation}\label{eq:contradict-Tp}
\sum_{i: b_i = 0} d_{\TV} (\Pi_{b \oplus e_i}, \Pi_b) = kp, 
\end{equation}
where $p < \frac{1-c}{2}$.  Let $p_i = d_\TV  \left(\Pi_{b \oplus e_i}, \Pi_b\right)$ for every player $i\in [T]$. \autoref{lem:not-seen} implies that the protocol can be equivalently implemented such that  if the other players receive $b_{-i}$, player $i$ only looks at their input with probability $p_i$.  If the player $i$ does not look at their bit, then their message $M_i$ is independent of their input bit. 
Let $Y_i$ denote the indicator random variable for the event that player $i$ looks at their input in this equivalent protocol. 

It follows from our assumption \eqref{eq:contradict-Tp} that if the input is $b$,  then   $\E { \sum_{ i : b_i = 0} Y_i}  = \sum_i p_i = kp$. By the definition of $Y_i$, if for any set $S$, $Y_i = 0$ for all $i\in S$, then all players in $S$ do not look at their input bits. 
Let $E_S$ denotes the event that $Y_i =0$ for all $i\in S$, for some $S \subseteq \{i: b_i = 0\}$.  
Then since the players in $S$ do not look at their input bits,
\[d_\TV (\Pi_{b\oplus e_S} | E_S, \Pi_b | E_S)=0.\]
In particular, using this and the law of total probability, we get that 
\begin{align}\label{eqn:single-tv}
    d_\TV (\Pi_{ b \oplus e_S}  , \Pi_b) &= \Pr (E_S) \cdot d_\TV (\Pi_{b\oplus e_S} | E_S, \Pi_b | E_S) + \Pr (\overline{E_S})  \cdot  d_\TV (\Pi_{b\oplus e_S} |\overline E_S, \Pi_b | \overline E_S) \nonumber \\ & \leq \Pr (\overline{E_S}).
\end{align}
By our assumption, $\E{\sum_{i: b_i=0} Y_i } = kp $ for $p< \frac{1-c}{2}$. 
Applying \autoref{lem:large-set}, we obtain that there exists a set $S\subseteq  \{i:b_i = 0\}$ with $|S|= ck$ such that $\Pr(E_S) \geq  e^{{-k/\gamma_c-1}}$. For any $k < {( T  -2)}{\gamma_c}<T-2$, we have $\Pr (E_S)  > e\delta $, and so $\Pr (\overline{E_S}) < 1-e\delta $. By Eqn.\ \eqref{eqn:single-tv}, $d_\TV (\Pi_{ b \oplus e_S}  , \Pi_b)  < 1-e\delta$. Observe that  $b \oplus e_S$  differs from $b$ by having $|S|=ck$ more $1$'s; and they have same value  at all other coordinates.
Recall that in a typical single-coordinate YES instance, there are $T-M$ number of $1$'s, which is $T/2 - M$ more than a typical NO instance. 
Now   suppose this gap $T/2 - M <  ck$; then solving $\biasdetect_1$ is at most as hard as distinguishing $b$ and $b\oplus e_S$.
Hence,  if we  choose $c'$ such that  $ M = c'T > T/2 - ck$, then the protocol $\Pi$ fails  with probability greater than $\delta$. This is a contradiction. 

Thus, for any $b$ such that $ck = c\cdot | \{ i : b_i = 0\}|>T/2 - M$,  
\begin{equation*}\label{eq:true-Tp}
\sum_{i: b_i = 0} d_{\TV} (\Pi_{b \oplus e_i}, \Pi_b) \geq\Omega \left( \frac{(1-c)T}{2} \right). 
\end{equation*}
From \eqref{eqn:final-line} and Jensen's inequality, 
\begin{align*}
      I(\Pi; C_1,\cdots, C_T)  \geq \Omega \left((1-c)^2 T\right).
\end{align*}
This finishes the proof.
\end{proof}

\end{document}